\documentclass[prd,twocolumn,superscriptaddress,floatfix,amsmath,amssymb,amsfonts,longbibliography,nofootinbib]{revtex4-2}
\usepackage{float}

\usepackage{comment}
\usepackage[normalem]{ulem}
\usepackage[english]{babel}
\usepackage{graphicx}% Include figure files
\usepackage{dcolumn}% Align table columns on decimal point
\usepackage{bm}% bold math
\usepackage{blindtext}
\usepackage{verbatim}
\usepackage{mathrsfs}
\usepackage{musicography}
\usepackage{amsmath}
\usepackage{dsfont}
\usepackage{mathrsfs}
\usepackage{amssymb}
\usepackage{amsthm}
\usepackage{cancel}
\usepackage{physics}
\usepackage{epstopdf}
\usepackage{mathtools}
\usepackage{color}
\usepackage[usenames,dvipsnames]{pstricks}
\usepackage{epsfig}
\usepackage{pst-grad} % For gradients
\usepackage{pst-plot} % For axes
\usepackage{hyperref}
\usepackage{verbatim}
\usepackage{afterpage}
\usepackage{sidecap}
\hypersetup{
    colorlinks=true,
    linkcolor=blue,
    filecolor=red,      
    urlcolor=cyan,
}
\usepackage{tikzsymbols}
\newcommand{\ii}{\mathrm{i}}

\newtheorem{theorem}{Theorem}
\newtheorem{lemma}{Lemma}
\newtheorem{proposition}{Proposition}
\newtheorem{corollary}{Corollary}

\newcommand{\kittyket}
{\ket{\raisebox{-.43ex}{\SchrodingersCat{1}}\!\!}}

\newcommand{\kittybra}
{\bra{\raisebox{-.43ex}{\SchrodingersCat{1}}\!\!}}

\DeclareMathOperator*{\sumint}{%
\mathchoice%
  {\ooalign{$\displaystyle\sum$\cr\hidewidth$\displaystyle\int$\hidewidth\cr}}
  {\ooalign{\raisebox{.14\height}{\scalebox{.7}{$\textstyle\sum$}}\cr\hidewidth$\textstyle\int$\hidewidth\cr}}
  {\ooalign{\raisebox{.2\height}{\scalebox{.6}{$\scriptstyle\sum$}}\cr$\scriptstyle\int$\cr}}
  {\ooalign{\raisebox{.2\height}{\scalebox{.6}{$\scriptstyle\sum$}}\cr$\scriptstyle\int$\cr}}
}

\begin{document}

\title{Semiclassical gravity beyond coherent states}

\author{Shahnewaz Ahmed}
\email{ahmed.shahnewaz16@gmail.com}
\affiliation{Perimeter Institute for Theoretical Physics, Waterloo, Ontario, N2L 2Y5, Canada}
\affiliation{Department of Physics and Astronomy, University of Waterloo, Waterloo, ON N2L 3G1, Canada}

\author{Caroline Lima}
\email{clima@perimeterinstitute.ca}
\affiliation{Perimeter Institute for Theoretical Physics, Waterloo, Ontario, N2L 2Y5, Canada}
\affiliation{Department of Physics and Astronomy, University of Waterloo, Waterloo, ON N2L 3G1, Canada}
\affiliation{Institute for Quantum Computing, University of Waterloo, Waterloo, Ontario, N2L 3G1, Canada}

\author{Eduardo Mart\'{i}n-Mart\'{i}nez}
\email{emartinmartinez@uwaterloo.ca}
\affiliation{Perimeter Institute for Theoretical Physics, Waterloo, Ontario, N2L 2Y5, Canada}
\affiliation{Institute for Quantum Computing, University of Waterloo, Waterloo, Ontario, N2L 3G1, Canada}
\affiliation{Department of Applied Mathematics, University of Waterloo, Waterloo, Ontario, N2L 3G1, Canada}

\date{\today}% It is always \today, today,
             %  but any date may be explicitly specified

\begin{abstract}
We show that it is possible to still use semiclassical gravity together with quantum field theory beyond the regimes where the field state is coherent. In particular, we identify families of cat states (superposition of almost-distinguishable coherent states that have very non-classical features) for which the gravitational backreaction can be  modeled by semiclassical gravity.
\end{abstract}

\maketitle

%\tableofcontents

\section{\label{sec:introduction} Introduction}

%We are yet to have a final theory for quantum gravity, but many advancements towards this direction have been done. 
One condition that any quantum gravity theory has to satisfy is that in the low energy limit, it has to match all well-established results in general relativity and quantum field theory.
Because of that, it not unreasonable to consider a bottom-up perspective  to quantum gravity where general relativity and semiclassical gravity are taken to the most quantum possible regimes that they can still model.
Under this perspective, one of our most powerful tools is quantum field theory in curved spacetimes~\cite{Birrell,WaldBookQFTCS,fulling1989aspects,NavarroSalas}, a framework that analyzes how a classical curved spacetime affects the behaviour of  quantum matter fields. Complementary to this is the study of the gravitational backreaction, that is, how quantum matter fields affect the dynamics and structure of a classical spacetime. Our first attempts to account for this backreaction yielded the theory known as \textit{semiclassical gravity}~\cite{hu2008stochastic, hubook}. In this framework, one considers the Einstein field equations but instead of taking the classical stress-energy tensor as the matter source, one takes the expectation value of its quantum version in the state of the matter field.

As expected for a low energy limit theory, semiclassical gravity is not a good approximation for regimes where the quantum nature of gravity is important. But it breaks down even far away from that regime: Since it takes into account only the expectation value of the stress-energy tensor of the quantum matter fields, when considering quantum states for which the variance or higher moments of this tensor are large enough, we expect semiclassical gravity to not suffice. Indeed, one way to formalize this intuition is using the stochastic gravity framework, a step above semiclassical gravity~\cite{hu2008stochastic, hubook}. This theory still considers the gravitational field to be classical, but instead of implementing the backreaction effect of the matter field via the average of its stress-energy tensor only, it also takes into account its fluctuations. Because this theory accounts for the quantum effects of the variance of the stress-energy tensor, it can be used as a guide to know when semiclassical gravity is a good approximation or not. In a few words, conclusions based on the semiclassical gravity framework are trustworthy in the cases where the variance (and higher moments) of the stress-energy tensor is small in the analyzed quantum state, so that stochastic gravity is not necessary \cite{Kuo93}.

Because of the simplicity of the semiclassical gravity framework, one may be tempted to apply it as widely as possible, and therefore it is relevant to find what families of states of quantum fields allow for its use. In other words, in order to delimit the setups for which semiclassical gravity is a good approximation, one needs to study families of quantum states for the matter fields so that the fluctuations of the stress-energy tensor are small enough. In 1993, Kuo and Ford \cite{Kuo93} studied this problem and one of their conclusions was that coherent states are adequate for this framework. This is not surprising, since the behaviour of coherent states is in all regards rather classical~\cite{Holevo2011}. Semiclasical gravity, however, would be not particularly interesting to model the backreaction of very quantum states, such as large superpositions of distributions of matter. 

Scenarios where one would like to consider the gravitational backreaction of quantum superpositions of classical states are quite relevant, for example, in  experiments where one wants to observe gravity induced entanglement to draw conclusions about the quantum nature of gravity~\cite{GravInd_spinentang,GravInd_MarlettoVedral,GravInd_ChristodoulouRovelli,GravInd_ChristodoulouRovelli2,GravInd_CarneyMullerTaylor,GravInd_PedernalesStreltsovPlenio,GravInd_Rudolph,GravInd_Hofer,GravInd_EduTales}. For this reason, it would be interesting to investigate to what extent there exist families of states such that at the same time 1) they display strong quantum features and 2) they are still within the limits of the semiclassical gravity framework. If such families of states existed, we could use them to effectively describe the backreaction of the structure of spacetime to matter states that are non-classical (unlike the coherent states). Thus perhaps helping to distinguish---in some hypothetical experiment of, for example, gravity induced entanglement---whether stochastic extensions of semiclassical gravity could be falsified when gravity is sourced by matter in non-classical states. %by account for quantum effects in the fabric of spacetime, but without abusing the approximation. 
 In this paper we will show that coherent cat states\footnote{I.e., quantum superpositions of distinguishable coherent states}, known to have quantum behaviour, are suitable for the semiclassical gravity framework.

Our paper is organized as follows. In Section~\ref{sec:setup} we describe the physical setup that will be used throughout this work. In Section \ref{sec:states} we introduce coherent and coherent cat-like states, the latter being the main object of our investigation. In 
Section \ref{sec:semiclassical}, we introduce the semiclassical gravity framework in more detail and state the criterion proposed by Kuo and Ford~\cite{Kuo93} for the validity of semiclassical gravity. In Section \ref{sec:results} we present our results for the validity of semiclassical gravity for coherent cat states, based on (an extension to higher moments) of the Kuo-Ford criterion. In Section \ref{sec:conclusions} we discuss our results. We use natural units $c=\hbar=1$ throughout the paper. 
    
\section{\label{sec:setup} Setup}

Consider a scalar field\footnote{The choice of a scalar field is taken for simplicity and convenience but our results extrapolate to more general setups} $\phi$ in a $(d+1)$-dimensional Lorentzian globally hyperbolic manifold $(\mathcal{M}, g_{\mu \nu})$ with mostly positive signature. Let us consider the Hilbert-Einstein action coupling gravity with the scalar field with an arbitrary coupling to the scalar curvature:
\begin{equation}
    S = -\int \dd^{d+1} \mathsf{x}\sqrt{-g} \ \frac{1}{2} (\partial_{\mu} \phi \partial^{\mu} \phi +m^2\phi^2 + \zeta R \phi^2). 
    \label{Action}
\end{equation}
Here, $\hbar = c = 1$, $g$ is the determinant of the metric $g_{\mu \nu}$, $m$ is the mass of the scalar field, $R$ is the Ricci scalar, and $\zeta$ is some coupling constant. %There are two very common choices for $\zeta$ in the literature; the first one is $\zeta=0$, known as the minimal coupling, and the second one is  $\zeta=\frac{d-2}{4(d-1)}$, regarded as the conformal coupling \cite{weinberg1972gravitation}. 
The corresponding equation of motion for the scalar field can be found by Hamilton's principle and it is given by:
\begin{equation}
    -\Box \phi + m^2 \phi + \zeta R \phi = 0,
    \label{KGeqn1}
\end{equation} 
where $\Box = \nabla_{\mu}\nabla^{\mu} = ({\sqrt{-g}})^{-1}\partial_{\mu}(\sqrt{-g} \ \partial^{\mu})$ is denoted as the  d'Alembertian operator. The global hyperbolicity of the manifold allows us to define an \textit{inner product} between a pair of solutions $\phi_1, \phi_2$ of the equation \eqref{KGeqn1} as follows:
\begin{equation}
    \langle \phi_1,\phi_2 \rangle =  \ii \int_{\Sigma} \left(\phi_2^* \partial_\mu \phi_1 - (\partial _{\mu} \phi_2^*) \phi_1 \right) \text{d}\Omega^{\mu},
    \label{innerprod}
\end{equation}
where $d\Omega^{\mu} = d\Omega \ n^{\mu}$, $d\Omega$ is the \textit{volume} element of the spacelike $d$ dimensional Cauchy hypersurface $\Sigma \subset \mathcal{M}$, and $n^{\mu}$ is a timelike unit vector perpendicular to $\Sigma$. Since this product is independent of the choice of the hypersurface~\cite{Birrell,ford2002d3}, it is a suitable candidate for defining the inner product over the space of the solutions $\mathcal{F}(\mathcal{M})$ of  equation \eqref{KGeqn1}, which is a subspace of smooth functions over the manifold $\mathcal{M}$.

To quantize the field, we will make the additional assumption that the spacetime is stationary, or at least asymptotically stationary. This in order to avoid complications regarding non-unitary equivalence of quantization schemes~\cite{fulling1989aspects,NavarroSalas}. In the next step, we assume $\Phi=\{\phi^+_i,\phi^-_i\}$ is the complete set of mode solutions (positive and negative norm respectively with respect to the stationary time direction) of equation \eqref{KGeqn1}, with the following orthogonality condition:
\begin{equation}
    \langle \phi_i^+,\phi_j^+ \rangle = \delta_{ij}, \quad \langle \phi^-_i,\phi^-_j \rangle = -\delta_{ij}, \quad \text{and} \quad \langle \phi_i^+,\phi^-_j \rangle = 0.
    \label{OrthoCondt}
\end{equation}

Here, the index $i$ stands for all the possible labels that can be associated with the solutions.  Since the spectrum of the differential operator in equation~\eqref{KGeqn1} can be both continuous and/or discrete, these labels could represent either discrete indices or continuous parameters. For example in the case of Minkowski spacetime, we could use a basis of plane waves labeled by their wave vector; but for Anti-de Sitter (AdS) spacetime, we require continuous and discrete labels to identify all the solutions. Moreover,  the delta functions in the set of equations \eqref{OrthoCondt} will be different (Kronecker delta or Dirac delta) depending on the nature of label $i$. Therefore, any $\phi \in \mathcal{F}(\mathcal{M})$ can be represented as
\begin{equation}
    \phi = \sumint_i \left( a^+_i \ \phi_i^+ + a^-_i \ \phi_i^- \right),
    \label{PhiClassical}
\end{equation}
with the following definitions for the expansion coefficients:
\begin{equation}
    a^+_i = \langle \phi,\phi_i^+ \rangle \in \mathbb{C}, \quad \text{and} \quad a^-_i = -\langle \phi,\phi^-_i \rangle \in \mathbb{C}.
\end{equation}

Let a Cauchy surface $\Sigma_t$ be a constant time-slice at some suitable time coordinate $t$ and $n^{\mu}$ be  a unit normal vector to it. The existence of such a surface can be ensured by the stationary property of the spacetime and the whole spacetime can be foliated with these surfaces. Along the time direction, the derivative of the scalar field is $\dot \phi =  n^{\mu}\partial_{\mu} \phi$. Using the Lagrangian density from the action in equation \eqref{Action} (let us call it $\mathcal{L}$) we can define the canonical momentum $\Pi = {\delta \mathcal{L}}/ {\delta \dot \phi}$ corresponding to the scalar field.  

Next, we apply the quantization map $\phi \to \hat \phi, \ \Pi \to\hat{\Pi}$ and  impose the equal time (at time $t$) commutation relations
\begin{equation}
    [\hat \phi(\mathbf{x},t),\hat \Pi(\mathbf{x}, t)] = \ii \hat \openone \, \delta(\mathbf{x}-\mathbf{x'}),
\end{equation}
where $\hat{\openone}$ is the identity in the operator algebra,  $\mathbf{x}, \mathbf{x'} \in \Sigma_t$, and the delta function is defined using the following relation
\begin{equation}
    \int_{\Sigma_t} \delta(\mathbf{x}-\mathbf{x'}) \ \text{d}\Omega = 1.
\end{equation}
This quantization procedure will change the classical observable $\phi$ in equation \eqref{PhiClassical} to the quantum mechanical operator 
\begin{equation}
    \hat \phi = \sumint_i \left( \hat a_i \ \phi_i^+ + \hat a^{\dagger}_i \ \phi_i^- \right),
    \label{PhiQuantum}
\end{equation}
with $[\hat a_i, \hat a_j^{\dagger}] = \delta_{ij} \hat{\openone}$. Hence, using the annihilation operator $\hat a_i$ we can define a state $\ket{0}$ which can be annihilated by all the annihilation operators. Moreover, a Hilbert space $\mathcal{H}$ can be constructed using Fock states by applying the creation operator $a^{\dagger}_i$ on the vacuum state $\ket{0}$. This procedure works well for Minkowski spacetime and generates a Poincar\'e invariant ground state; however, for general curved spacetimes, such a state is not unique. In lieu of spacetime symmetries we do not have a privileged `vacuum' state and hence a natural Fock basis. We can always find another set of mode solutions $\Phi'$ different from $\Phi$, which can span $\mathcal{F}(\mathcal{M})$ and produce a different, non-equivalent, vacuum state $\ket{0'} \neq \ket{0}$. 

This creates an ambiguity in the definition of ``particle" found in QFT in flat spacetime. This nonuniqueness is behind phenomena like, e.g., particle creation in the early universe~\cite{sexl1969production, GibbonsHawking} or the Unruh effect~\cite{Fulling1973, Davies1975,Unruh,UnruhWald1984,Takagi,CrispinoMatsas}. With this possible ambiguity in mind, it is still possible to manipulate quantum fields in curved backgrounds~\cite{hollands2015quantum}. Within this framework, many productive results have been obtained, including Hawking radiation~\cite{hawking1974black, hawking1975particle, wald1975particle}, the Casimir effect in curved spacetime \cite{ford1975quantum, ford1976quantum}, as well as applications to cosmology \cite{mukhanov2005physical, dodelson2020modern}.

Within the  framework of quantum fields in curved spacetime, the stress-energy tensor $T^{\mu \nu} \coloneqq (2/\sqrt{-g}) \  (\delta S/ \delta g_{\mu \nu})$ plays a prominent role, since it is stress-energy that sources gravity in general relativity. For the action $S$ in equation \eqref{Action}, it becomes
\begin{align}
    T^{\mu \nu} =& \nabla^{\mu} \phi \nabla^{\nu} \phi - \frac{1}{2}g^{\mu \nu}\left(\nabla^{\rho} \phi \nabla_{\rho} \phi + m^2 \phi^2\right) \nonumber\\
    &+ \zeta \left(G^{\mu \nu} - g^{\mu \nu} \Box + \nabla^{\mu} \nabla^{\nu} \right) \phi^2,
    \label{eq:Tmunu}
\end{align}
where $G_{\mu \nu}$ is the Einstein tensor \cite{callan1970new}. After the quantization procedure $T^{\mu \nu}[\phi] \to \hat T^{\mu \nu}[\hat{\phi}]$ (with Weyl ordering), the operator $\hat T^{\mu \nu}[\hat{\phi}]$ becomes ill-defined because the mathematically well-behaved operator $\hat \phi(\mathsf{x}), \mathsf{x} \in \mathcal{M}$, is a distribution (must be integrated by modulating with a function $f(\mathsf{x})$ over compact support) and non-linear operations (like multiplication) involving distributions are not well-defined. To remedy the situation, we need to regularize and renormalize the expectation values involving $\hat T^{\mu \nu}$. There are several methods available to do so \cite{hu2008stochastic}, and in the next section we make use of the most common one (normal ordering) to describe the theoretical framework of semiclassical gravity.

\section{\label{sec:states} Coherent and cat-like states}

%%%%%%%%%%%%%%%%%%%%%%%%%%%%%
%%%%%%%%%%%%%%%%%%%%%%%%%%%%%
\subsection{Coherent states}
Coherent states are sometimes regarded as the boundary between classical and quantum states of a quantum field~\cite{robertmonique}. This is related to the fact that coherent states have Poissonian distributions in the `number of excitations' observable. For this reason they are typically considered as good models for macroscopic coherent light, e.g., laser radiation trapped in an optical cavity. Consistently, those states do not display genuine quantum behaviour.  Indeed, their Wigner functions are Gaussian probability distributions and, as such, they can in principle be modeled within classical field theory (see, e.g.,~\cite{Edulectnotes}). Coherent states are defined by being right-acting eigenstates of the annihilator operators:
\begin{eqnarray}
    a_i \ket{\alpha}  = \alpha_i \ket{\alpha},
    \label{coherdef}
\end{eqnarray}
where $\alpha_i$ is a complex-valued distribution (that we will call coherent amplitude). These states can be fully characterized by applying a (phase space) displacement operator to some Hadamard state $\ket{0}$: 
\begin{equation}
    \ket{\alpha} = \hat D(\alpha) \ket{0},
    \label{dispfromvac}
\end{equation} 
where 
\begin{equation}
    \hat D(\alpha) \coloneqq \exp\left(\hat{A}^{\dagger}(\alpha) - \hat{A}(\alpha)\right) ,
    \label{displacmentOp}
\end{equation}
and we defined
\begin{align}
    \hat{A}(\alpha) = \sumint_i \alpha^{*}_{i} \, \hat a_i, \,\quad \text{and}\quad\, 
    \hat{A}^{\dagger}(\alpha) = \sumint_i \alpha_{i} \, \hat a^{\dagger}_i, 
    \label{AAdef}
\end{align}
where $\alpha_i \coloneqq \langle \alpha, \phi^+_i \rangle \in \mathbb{C}$ and $\alpha_i^*$ is the complex conjugate of $\alpha_i$. % Note that here $\alpha$ contains all the possible continuous and discrete indices present in the representation of the Hilbert space $\mathfrak{h}$.  Following \cite{robert2021coherent}, we define for every $\alpha \in \mathfrak{h}$, a Weyl translation operator (also known as displacement operator)
%\begin{eqnarray}
%    \hat D(\alpha) = \exp\left( \int %\dd^3\bm{k} \left[\alpha(\bm{k})\hat %a_{\bm{k}} - \alpha^*(\bm{k})\hat %a^{\dagger}_{\bm{k}}  \right] \right),
%\end{eqnarray}
%where $\alpha(\bm k)$ is a distribution to derive explicit formula for coherent state \cite{truax1985baker}
%Phys. Rev. D 31, 1988 (1985).
%\begin{eqnarray}
%    \ket{\alpha} = \hat D(\alpha) \ket{0}.
%\end{eqnarray}
The relationship between equations ~\eqref{coherdef} and ~\eqref{dispfromvac}  can be easily proven by using the commutation relation
\begin{eqnarray}
\label{eqn:ccrelation}
    [ \hat a_i, \hat D(\alpha)] = \alpha_i \hat D(\alpha),
\end{eqnarray}
which is derived in Appendix~\ref{appendixA}.
%%%%%%%%%%%%%%%%%%%%%%%%
%Note that here $\alpha$ contains all the possible continuous and discrete indices present in the representation of the Hilbert space $\mathfrak{h}$. 
%Following \cite{robert2021coherent}, we define for every $\alpha$, a Weyl translation operator (also known as displacement operator)
%\begin{equation}
 %   \hat D(\alpha) \coloneqq \exp\left(\hat{A}^{\dagger}(\alpha) - \hat{A}(\alpha)\right) \in \mathcal{B}(\mathcal{H}).
 %   \label{displacmentOp}
%\end{equation}
%%%%%%%%%%%%%%%%%%%%%%%%

%These states are relevant in many application specially in quantum optics. \cite{scully1999quantum}
%M. O. Scully and M. S. Zubairy, Quantum optics
%We can produce a same setup for studying entanglement behaviour  using two coherent states $\ket{\alpha}$ and $\ket{-\alpha}$. We will assume by appropriate choice of distribution $\alpha(k)$ we can make these two states orthonormal. This will allow us to represent $\ket{\alpha}$ and $\ket{-\alpha}$ as $\ket{01}$ and $\ket{10}$ respectively. 
%%%%%%%%%%%%%%%%%%%%%%%%%%%%%
%%%%%%%%%%%%%%%%%%%%%%%%%%%%%

\subsection{Cat-like states}

In contrast to the coherent states defined above, cat-like states are built out of simple superpositions of coherent states:
\begin{equation}
\big|\Psi_{_{\SchrodingersCat{1}}}\big\rangle=a\ket{\alpha}+b\ket{\beta},
\end{equation}
where $a, b \in \mathbb{C}$ are such that the state is normalized.

Coherent states are not orthogonal, but one can choose a cat state so that if the magnitude of the coherent amplitude $|\alpha|^2= \langle \alpha, \alpha \rangle$ is large enough, the inner product between the two states becomes effectively negligible\footnote{The inner product between two coherent states is given by \cite{robertmonique}
    \begin{align*}
        \braket{\alpha}{\beta} = \exp( \ii  \, \mathfrak{Im}\{\langle \alpha, \beta \rangle\}) \exp(-\frac{|\beta-\alpha|^2}{2}).  
        \label{eq:alphabetasup}
    \end{align*}
}. This is the case for the canon cat state\footnote{Namely, $
    \braket{\alpha}{-\alpha} = e^{-2|\alpha|}$.}
\begin{equation} \label{eq:catt}
    \kittyket=a\ket{\alpha}+b\ket{-\alpha},
\end{equation}
where the inner product between $\ket{\alpha}$ and $\ket{-\alpha}$ is exponentially suppressed with the coherent amplitude.

Unlike the coherent states, cat states are highly non-classical. Intuitively, this may not be surprising since, like Schr\"odinger cat, they are the superposition of two (approximately) distinguishable macroscopic states of the field. More rigorously, the cat states' Wigner function is not a probability distribution since it takes on negative values. This is the hallmark of non-classical behaviour.

In summary, cat states are particularly interesting because a) they can easily be produced in the lab and b) they display genuinely quantum behaviour.

% \shah{Put inner product}

\section{\label{sec:semiclassical} Semiclassical Gravity}

Semiclassical Gravity is perhaps the simplest way to account for the effect of  quantum matter on the classical gravitational field within the context of quantum field theory in curved spacetimes~\cite{hubook}. 

To introduce semiclassical gravity, we begin from Einstein equations, given by
\begin{equation}\label{eq:einstein}
    G_{\mu\nu} = 8\pi G_N T_{\mu\nu},
\end{equation}
where $G_{\mu\nu}$
is the Einstein tensor, and $T_{\mu\nu}$ is the stress-energy tensor for matter.
Notice that this equation relates the geometry of the spacetime, on the LHS, with properties of matter, on the RHS. Therefore, if we want to introduce quantum matter, we need to upgrade the matter observables (including the stress-energy tensor) to self-adjoint operators on the matter's Hilbert space. However, as of today, we do not know how to map dynamical functions from the LHS  (the gravity sector) to self-adjoint operators, so we cannot just ``put hats'' on the RHS of Einstein equations. Semiclassical gravity addresses this by taking the expectation value of the stress-energy tensor and thus giving a set of dynamical equations for the spacetime geometry, at least on average. In other words, instead of equation~\eqref{eq:einstein}, we write
\begin{equation}\label{eq:semiclassicaleinstein}
    G_{\mu\nu} = 8\pi G_N \langle\hat{T}_{\mu\nu}\rangle.
\end{equation}
These are known as the semiclassical Einstein Equations.

The first problem one encounters when doing this is that $\langle \hat T_{\mu\nu}\rangle$ is divergent. However, it can be renormalized. The simplest way is to use the empirically obsverved fact that zero point energy does not seem to have a measurable impact in the curvature of spacetime. That is, when the quantum fields are in vacuum, observations of the gravitational field are compatible with flat spacetime. Therefore we can subtract the (divergent) Minkowski vacuum expectation from the stress-energy density. This mechanism works because the divergence in the expectation of the stress-energy density for a scalar field $\hat\phi$ happens at the coincidence limit, and for a regular enough spacetime, the (UV-divergent) coincidence limit coincides with the Minkowski vacuum expectation value~\cite{Birrell}.

To implement this we define the normal ordered stress-energy tensor as
\begin{equation}
    :\!\hat T_{\mu\nu}\!: \,\,\,\coloneqq\,\,\, \hat T_{\mu\nu} - \expval{\hat T_{\mu\nu}}_0,
\end{equation}
where $\langle\hat T_{\mu\nu}\rangle_0 = \bra{0} \hat T_{\mu\nu}\ket{0}$. We define the more general normal ordering operation for polynomials of the field amplitude (and its derivatives) as its usual operational definition~\cite{Weinberg}: the normal-ordered expectation is the expectation value of the operators once annihilation operators are placed on the right side and creation operators on the left side. In other words, when we take expectation values of powers of the stress-energy tensor---for example, to build the variance of the stress-energy, and its higher statistical moments---the expectation value is understood to be taken as a renormalized (by normal order) expectation. This is akin to subtracting all the divergent terms that come from vacuum expectations of the stress-energy density.

It is well-known that semiclassical gravity is just an effective theory. Even though there are no experiments yet that have found a regime where it is not applicable (although this may change in the next decades~\cite{GravInd_spinentang,GravInd_MarlettoVedral,GravInd_ChristodoulouRovelli,GravInd_ChristodoulouRovelli2,GravInd_CarneyMullerTaylor,GravInd_PedernalesStreltsovPlenio,GravInd_Rudolph,GravInd_Hofer,GravInd_EduTales}), we know that there is no consistent way of coupling classical fields to quantum fields~\cite{Terno, BarceloGaray} and, therefore, semiclassical quantum gravity is just an effective theory whose regime of applicability is limited.

The question of how limited this framework is was first addressed in~\cite{Kuo93,Ford00}. The argument in those papers is quite simple and can be summarized as follows: as long as the fluctuations of the stress-energy density do not dominate over its expectation value, semiclassical gravity is applicable. To quantify this, the following estimator was proposed~\cite{Kuo93}:
\begin{equation}
    \Delta_{\mu\nu\lambda\rho}{(\mathsf{x},\mathsf{x}')} \!=\! \left| \frac{\langle :\!\hat T_{\mu \nu}(\mathsf{x}) \hat T_{\lambda \rho}(\mathsf{x}')\!: \rangle \!-\! \langle :\!\hat T_{\mu \nu}(\mathsf{x})\!: \rangle \langle \hat T_{\lambda \rho}(\mathsf{x}')\!: \rangle }{\langle :\!\hat T_{\mu \nu}(\mathsf{x}) \hat T_{\lambda \rho}(\mathsf{x}')\!: \rangle} \right|.
    \label{eq:KuoFord}
\end{equation}
This estimator is the ratio between the variance of the stress-energy tensor and the expectation value of its square. As per the argument above, if this estimator is $\Delta_{\mu\nu\lambda\rho}(\mathsf{x},\mathsf{x}')\ll1$ for all $\mathsf{x},\mathsf{x}'$, then  we are within the approximation of semiclassical gravity. In the original work~\cite{Kuo93} it was found that for coherent states, this condition is fulfilled. However, if semiclassical gravity as a dynamical framework for spacetime were to work only for coherent states, one could question how much of the quantum nature of matter would be captured by this model. Indeed, as discussed above, coherent states are good models for macroscopic classical states for matter, so if they are the only safe choice for semiclassical gravity, this would put into question why use a quantum treatment for matter to begin with. However, as we will discuss, this condition (i.e., $\Delta_{\mu\nu\lambda\rho}(\mathsf{x},\mathsf{x}')\ll1$ for all $\mathsf{x},\mathsf{x}'$) is also satisfied for families of matter states that display strongly quantum behaviour.

This estimator is particularly useful for Gaussian states, since for those states all statistical moments of quadratic observables are functions of the second and first moments, so that satisfying the Kuo-Ford criterion guarantees that the state gravitates approximately semiclassically. 
However, superpositions of Gaussian states are not Gaussian. Hence, making sure that the Kuo-Ford  estimator~\eqref{eq:KuoFord} is small does not guarantee that semiclassical gravity is applicable. This is because, even if variance is not relevant, higher moments can be non-negligible when compared to the magnitude of the expectation of the renormalized stress-energy tensor.

As we will see, similar estimators to equation \eqref{eq:KuoFord} can be built for the higher moments of the renormalized stress-energy density. We will need to take them into account to decide whether some particular choices of cat states may still gravitate semiclassically. This is what we will do in the next section.

\section{\label{sec:results} Results}

The goal of this section is to show that cat states are suitable for semiclassical gravity. We will do so by showing not only that they satisfy the criterion proposed by Kuo and Ford \cite{Kuo93}, but also that these very non-classical states of a scalar field have vanishing symmetrized higher (regularized) moments of the stress-energy density.
That is, we want to show that for the gravitational backreaction of cat states, semiclassical gravity is still valid. In the following, we proceed by proving a series of statements that will lead to this result.
\begin{lemma}
    \label{lemma:phi}
    Let $\hat \phi(\mathsf x)$ be a scalar quantum field, $\mathfrak{D}_{\mathsf x}$ some linear differential operator, ${\mathfrak{h}}$ and $\,\bar{\mathfrak{h}}\,$ $\mathbb{C}$-linear vector spaces spanned by $\{\phi^+_i\}$ and $\{\phi^-_i\}$, respectively. Then,
    \begin{equation}
        \bra{\alpha}\mathfrak{D}_{\mathsf x} \hat\phi (\mathsf x)\ket{\beta} = \braket{\alpha}{\beta} \mathfrak{D}_{\mathsf x}(\beta(\mathsf x) + \bar{\alpha}(\mathsf x)), 
    \end{equation}
    where $\beta(\mathsf x) =\sumint_i \beta_i \phi_i^+(\mathsf x)\in \mathfrak{h}$ and   $\Bar{\alpha}(\mathsf x) = \sumint_i \alpha^*_i \phi_i^-(\mathsf x) \in \bar{\mathfrak{h}}$.
\end{lemma}

\begin{proof} This formula can be derived easily using linearity and equations \eqref{PhiQuantum} and \eqref{eqn:ccrelation}: 
    \begin{align}
    \bra{\alpha}\mathfrak{D}_{\mathsf x} \hat \phi(\mathsf x) \ket{\beta} &=  \sumint_i \bigg( \bra{\alpha} [\hat  a_i, \hat D(\beta)] \ket{0} \mathfrak{D}_{\mathsf x}\phi_i^+(\mathsf x)  \ \nonumber\\& \quad \quad -  \sumint_i  \bra{0}[\hat a^{\dagger}_i, \hat D^{\dagger}(\alpha)] \ket{\beta} \ \mathfrak{D}_{\mathsf x} \phi_i^-(\mathsf x) \bigg) \nonumber \\
    &= \braket{\alpha}{\beta} \sumint_i (\beta_i \, \mathfrak{D}_{\mathsf x}\phi_i^+ (\mathsf x)+ \alpha^*_i \, \mathfrak{D}_{\mathsf x} \phi_i^-(\mathsf x)) \nonumber \\ &= \braket{\alpha}{\beta} \mathfrak{D}_{\mathsf x}(\beta(\mathsf x) + \bar{\alpha}(\mathsf x)).
    \end{align}
\end{proof}
Notice that the displacement operator $\hat D(\beta)$ does not depend on $\mathsf x \in \mathcal{M}$; nevertheless, we will keep writing the coherent state as $\ket{\beta}$ where $\beta \in \mathfrak{h}$. Furthermore, for the next theorem, we will remove the direct $\mathsf x$ dependence in $\hat \phi(\mathsf x)$ to reduce clutter.

\begin{theorem} \label{thm:poly}
Let
\begin{align}
    &P[ \hat{\phi},  \nabla^{\mu} \hat{\phi}, \nabla^{\mu}\nabla^{\nu} \hat{\phi}, \Box \hat {\phi} , \cdots ; \nonumber\\
    & \quad \hat \phi', \, \nabla^{\mu'}\phi' \, ,\nabla^{\mu'}\nabla^{\nu'} \hat{\phi'}, \Box \hat {\phi'} , \cdots ;\cdots\, \nonumber]
\end{align}
be a polynomial. 
Then 
    \begin{align}
        &\bra{\alpha}:P[ \hat{\phi},  \nabla^{\mu} \hat{\phi}, \nabla^{\mu}\nabla^{\nu}\hat{\phi}, \Box \hat {\phi}  , \cdots ]:\ket{\beta} =\braket{\alpha}{\beta} \label{theor1} \\
        &\times
         P[ (\beta + \bar{\alpha}),  \nabla^{\mu} (\beta + \bar{\alpha}), \nabla^{\mu}\nabla^{\nu} (\beta + \bar{\alpha}), \Box (\beta + \bar{\alpha})  , \cdots ],\nonumber
    \end{align}
    where $:P:$ means that the operators in $P$ are in normal order, the primes represent evaluation at (in principle) different spacetime points and `` $\cdots$" represents higher order derivatives and other spacetime points.
\end{theorem}
\begin{proof}
    For notational simplicity in the proof, let us consider first the case where all operators in the polynomial are evaluated at the same spacetime point. At first, we want to split the field operator $\hat \psi$ into two different parts, i.e., annihilation and creation. To do so, we define two operators:
    \begin{eqnarray}
        \hat \psi^+ = \sumint_i \hat{a}^{\dagger}_i \phi^+_i,\, \quad \text{and} \quad  \hat \psi^- = \sumint_i \hat{a}_i \phi^-_i,
    \end{eqnarray}
    following the notation defined in Section \ref{sec:setup}.
    From Lemma \ref{lemma:phi} we see  that 
    \begin{equation}
        \bra{\alpha} \hat \psi^+  = \bra{\alpha} \bar{\alpha}  \quad \text{and} \quad \hat \psi^- \ket{\beta} =  \beta \ket{\beta}.
        \label{eq:eigpsi}
    \end{equation} 
    Moreover, from the properties of normal ordering, we can write down the following:
    \begin{align}
        &:P[ \hat{\phi},  \nabla^{\mu} \hat{\phi},\nabla^{\mu}\nabla^{\nu} \hat{\phi}, \Box \hat {\phi}, \cdots   ]: \nonumber \\  &= \sum_{k \ge 0} P^+_k[ \hat{\psi}^+,  \nabla^{\mu} \hat{\psi}^+,\nabla^{\mu}\nabla^{\nu} \hat{\psi}^+, \Box \hat {\psi}^+ , \cdots  ] \nonumber \\
        & \quad \quad \times P^-_k[ \hat{\psi}^-,  \nabla^{\mu} \hat{\psi}^-,\nabla^{\mu}\nabla^{\nu} \hat{\psi}^-, \Box \hat {\psi}^- , \cdots  ].
        \label{eq:PPPeqn}
    \end{align}
    Here, $P^+_k$ and $P^-_k$ are two different polynomials of $\hat{\psi}^+$ and $\hat{\psi}^-$ (and their derivatives), respectively.
    Since $\hat{\psi}^-$ and its derivatives commute with each other, 
    \begin{align}
        & P^-_k[\hat{\psi}^-,  \nabla^{\mu} \hat{\psi}^-,\nabla^{\mu}\nabla^{\nu} \hat{\psi}^-, \Box \hat {\psi}^-,  \cdots ] \nonumber \\
        &= \sum_{n_1, n_2, n_3, n_4 , \cdots }    \, ^{(k)}C^-_{n_1,n_2,n_3,n_4, \cdots } \left(\hat{\psi}^-\right)^{n_1}    \nonumber \\
        &\quad \quad \times \left(\nabla^{\mu} \hat{\psi}^-\right)^{n_2} \left(\nabla^{\mu}\nabla^{\nu} \hat{\psi}^-\right)^{n_3} \left(\Box  {\hat{\psi}^-}\right)^{n_4} \cdots ,
    \end{align}
    where each $n_i  \in \{0,1,\dots\}$ and the $^{(k)}C^-_{n_1,n_2,n_3,n_4, \cdots}$ coefficient can be a function of the metric tensor $g_{\mu \nu}$, scalar field mass $m$, coupling to curvature $\zeta$, etc. Next, we compute the action of the polynomial of operators   $P^-_k$ on $\ket{\beta}$. From equation \eqref{eq:eigpsi} and Lemma \ref{lemma:phi} we see that
    \begin{align}
        & P^-_k[\hat{\psi}^-,  \nabla^{\mu} \hat{\psi}^-,\nabla^{\mu}\nabla^{\nu} \hat{\psi}^-, \Box \hat {\psi}^-, \cdots] \ket{\beta} \nonumber\\
        =& \sum_{n_1, n_2, n_3, n_4, \cdots} \, ^{(k)}C^-_{n_1,n_2,n_3,n_4, \cdots} \left(\beta\right)^{n_1}  \nonumber \\
        &\quad \times \left(\nabla^{\mu} \beta\right)^{n_2} \left(\nabla^{\mu}\nabla^{\nu} \beta\right)^{n_3} \left(\Box  {\beta}\right)^{n_4} \cdots\ket{\beta} .
        \label{eq:Pminus}
    \end{align}
    The same argument goes for the polynomial $P^+_k$ acting (from the left) on $\ket\alpha$: 
    \begin{align}
        & \bra{\alpha} P^+_k[\hat{\psi}^+,  \nabla^{\mu} \hat{\psi}^+,\nabla^{\mu}\nabla^{\nu} \hat{\psi}^+, \Box \hat {\psi}^+, \cdots] \nonumber \\
        &= \bra{\alpha} \sum_{n_1, n_2, n_3, n_4} \, ^{(k)}C^+_{n_1,n_2,n_3,n_4} \left(\bar{\alpha}\right)^{n_1} \nonumber \\
        &\quad \times \left(\nabla^{\mu} \bar{\alpha}\right)^{n_2} \left(\nabla^{\mu}\nabla^{\nu} \bar{\alpha}\right)^{n_3} \left(\Box  {\bar{\alpha}}\right)^{n_4}\cdots
        \label{eq:Pplus}
    \end{align}
    From equations \eqref{eq:Pminus} and \eqref{eq:Pplus},  we see that all $\hat \psi^+$ and $\hat \psi^-$ are just getting replaced by $\bar{\alpha}$ and $\beta$ respectively within the polynomials. Hence,
    \begin{align}
        & \bra{\alpha} :P[ \hat{\phi},  \nabla^{\mu} \hat{\phi},\nabla^{\mu}\nabla^{\nu} \hat{\phi}, \Box \hat {\phi}  , \cdots]: \ket{\beta} \nonumber \\
        &= \braket{\alpha}{\beta} \sum_{k\geq0}   P^+_k[ \bar{\alpha},  \nabla^{\mu} \bar{\alpha},\nabla^{\mu}\nabla^{\nu} \bar{\alpha}, \Box \bar{\alpha}  , \cdots] \nonumber \\
        &\quad \quad \times P^-_k[ \beta,  \nabla^{\mu} \beta,\nabla^{\mu}\nabla^{\nu} \beta, \Box \beta, \cdots  ].
        \label{eq:alphbetaP}
    \end{align}
    We can rearrange the right-hand side of the equation \eqref{eq:alphbetaP} so that 
    the sum of all the polynomials can be written as a single polynomial, which is the polynomial $P$ after replacing $\hat \phi$ by $ \bar{\alpha}+ \beta$. This is possible because the right-hand side of the formula originated from $\hat \phi = \hat \psi^+ + \hat \psi^-$. Consequently, we derive Theorem \ref{thm:poly}. 
\end{proof}

The proof of the theorem would be the same if we were using the more general polynomial
\begin{align}
     P[\hat \phi, \nabla^{\mu}\phi, \, \cdots \, ; \hat \phi', \, \nabla^{\mu'}\phi' \, , \cdots \, ; \hat \phi'', \, \nabla^{\mu''}\phi'' \, , \cdots\, ; \cdots].
    \label{eq:generalizedthm}
\end{align}
The matrix element $\bra{\alpha} P \ket{\beta}$ can be easily found by substituting $\phi \to \bar \alpha + \beta$, $\phi'\to\bar \alpha' + \beta'$, and so on and so forth. 

\begin{corollary} \label{corollary1}
    \begin{align}
        &\bra{\alpha} :\hat{T}^{\mu \nu}: \ket{\beta} \nonumber \\
        &= \bigg(\nabla^{\mu} (\bar{\alpha}+ \beta) \nabla^{\nu} (\bar{\alpha}+ \beta) - \frac{1}{2}g^{\mu \nu}\nabla^{\rho} (\bar{\alpha}+ \beta) \nabla_{\rho} (\bar{\alpha}+ \beta) \nonumber \\
        &+ \left( \zeta \left(G^{\mu \nu}   + \nabla^{\mu} \nabla^{\nu}\right) - \frac{1}{2}g^{\mu \nu} (m^2+\zeta \Box)   \right) (\bar{\alpha}+ \beta)^2\bigg) \braket{\alpha}{\beta}.
        \label{eq:alphaTbeta}
    \end{align}
    % where $\bar{\alpha}+ \beta$ are the short form of $\bar{\alpha}(\mathsf x)+ \beta(\mathsf x)$.
\end{corollary}
\begin{proof}
    It is a straightforward consequence of applying Theorem \ref{thm:poly} over the operator version of the equation \eqref{eq:Tmunu}.
\end{proof}
    To reduce clutter, in the following we will write $\hat T'_{\sigma' \rho'}$ instead of $\hat T_{\sigma' \rho'}(\mathsf x')$. Thus, we will omit the explicit spacetime dependence of the operators.
    
\begin{corollary} \label{corollary2}
    \begin{equation}
        \bra{\alpha}: \hat T_{\mu \nu} \hat T'_{\sigma' \rho'} :\ket{\alpha} = \bra{\alpha} :\hat T_{\mu \nu}: \ket{\alpha} \bra{\alpha}: \hat T'_{\sigma' \rho'}: \ket{\alpha}.
        \label{eq:TTequation}
    \end{equation}
\end{corollary}
\begin{proof}
    By applying Theorem \ref{thm:poly}, we see that the polynomial inside normal ordering is in factorized form; hence, the right-hand side must be in the factored form.
\end{proof}

As a consequence, we derived that coherent states satisfy the Kuo-Ford criterion in globally hyperbolic stationary spacetime.

\begin{proposition}\label{proposition1}
Let $\kittyket = a \ket{\alpha} + b \ket{-\alpha}$ and $\braket{\alpha}{-\alpha} =\epsilon$, with $\epsilon\ll 1$. 
Then
\begin{eqnarray}
    \kittybra: \hat T_{\mu \nu}(\mathsf{x}):\kittyket &=& \bra{\alpha}:\hat T_{\mu \nu}:\ket{\alpha}+\mathcal{O}(\epsilon),
    \label{eq:claim4}\\
    \kittybra:\hat T_{\mu \nu} \hat T'_{\sigma' \rho'}:\kittyket &=& \bra{\alpha}:\hat T_{\mu \nu}\hat T'_{\sigma' \rho'}:\ket{\alpha}+\mathcal{O}(\epsilon), \, \label{eq:claim4b} 
\end{eqnarray}
and, in general, 
\begin{equation}
\kittybra: P[\hat T_{\mu \nu}]:\kittyket = \bra{\alpha}:\hat P[T_{\mu \nu}]:\ket{\alpha}+\mathcal{O}(\epsilon),
    \label{eq:claim4c}
\end{equation}
where $P[\hat T_{\mu \nu}]$ is any polynomial of $\hat T_{\mu\nu}$ (even if each instance of $\hat T_{\mu\nu}$ is evaluated at different spacetime points).
\end{proposition}

\textbf{Proof:}
% Notice, from \eqref{eq:alphaTbeta}, \eqref{eq:TTequation} and \eqref{theor1} that
\begin{align}
\kittybra:\hat T_{\mu \nu}:\kittyket &= |a|^2 \bra{\alpha}:\hat T_{\mu \nu}:\ket{\alpha} \nonumber \\ &+ |b|^2 \bra{-\alpha}:\hat T_{\mu \nu}:\ket{-\alpha} \nonumber \\ &+ \epsilon\,  a^*b \bra{\alpha}:\hat T_{\mu \nu}:\ket{-\alpha} \nonumber \\ &+ \epsilon\,  ab^* \bra{-\alpha}:\hat T_{\mu \nu}:\ket{\alpha}.
\label{eq:claim4general}
\end{align}
Now, from equation \eqref{eq:alphaTbeta} one can see that
\begin{align}
    \bra{-\alpha}:\hat T_{\mu \nu}:\ket{-\alpha} =\bra{\alpha}:\hat T_{\mu \nu}:\ket{\alpha}, 
    \label{eq:claim4eq3} \\
    \bra{-\alpha}:\hat T_{\mu \nu}:\ket{\alpha} =\bra{\alpha}:\hat T_{\mu \nu}:\ket{-\alpha},\label{eq:claim4eq4}
\end{align}
because $\overline{-\alpha} = \sumint -\alpha_i^* \phi^-_i = -\bar{\alpha}$, and we can recover equation \eqref{eq:claim4}. 

We now consider the case of quadratic terms involving stress-energy tensors. The key formula we need here is 
\begin{align}
    \bra{-\alpha}:\hat T_{\mu \nu} \hat T'_{\mu' \nu'} :\ket{-\alpha} =\bra{\alpha}:\hat T_{\mu \nu} \hat T'_{\mu' \nu'} :\ket{\alpha}, 
    \label{eq:claim4eq5} \\
    \bra{-\alpha}:\hat T_{\mu \nu} \hat T'_{\mu' \nu'} :\ket{\alpha} =\bra{\alpha}:\hat T_{\mu \nu} \hat T'_{\mu' \nu'} :\ket{-\alpha}.\label{eq:claim4eq6}
\end{align}
For illustration, we sketch the proof of the \eqref{eq:claim4eq6}; \eqref{eq:claim4eq5} follows the same proof. We introduce the notation:
\begin{equation} \label{eq:notationTmunu}
    T_{\mu \nu}[\alpha, \beta]  = \bra{\alpha}\hat T_{\mu \nu}\ket{\beta} /\braket{\alpha}{\beta}.
\end{equation}
In other words,  the expression for $T_{\mu \nu}[\alpha, \beta]$ is given by equation \eqref{eq:alphaTbeta} without the inner product. From equation \eqref{eq:alphaTbeta} and using the fact that $\braket{\alpha}{-\alpha}$ is real, one can check that
\begin{align}
    T_{\mu \nu}[\alpha, -\alpha]  = T_{\mu \nu}[-\alpha, \alpha] \\
    T_{\mu \nu}[-\alpha, -\alpha]  = T_{\mu \nu}[\alpha, \alpha].
\end{align}
Using Theorem \ref{thm:poly},
\begin{align}
   \nonumber \bra{-\alpha}:\hat T_{\mu \nu} \hat T'_{\mu' \nu'} :\ket{\alpha} &=  \braket{-\alpha}{\alpha}T_{\mu \nu}[-\alpha,\alpha]T'_{\mu' \nu'}[-\alpha,\alpha]  \\
 \nonumber   &= \braket{-\alpha}{\alpha} T_{\mu \nu}[\alpha,-\alpha]T'_{\mu' \nu'}[\alpha,-\alpha]  \\
    & = \bra{\alpha}:\hat T_{\mu \nu} \hat T'_{\mu' \nu'} :\ket{-\alpha}.
\end{align}
Using the equations \eqref{eq:claim4eq5} and \eqref{eq:claim4eq6}, we can employ the same steps used in the proof for $\kittybra:\hat T_{\mu \nu}:\kittyket$ to prove the result for $\kittybra:\hat T_{\mu \nu}\hat T'_{\mu' \nu'}:\kittyket$. The general case can be proven using similar steps for the polynomial $P[T_{\mu \nu}]$.

% The derivation of equation \eqref{eq:claim4b} is direct by using~\eqref{eq:TTequation} and Equations \eqref{eq:claim4eq3} and \eqref{eq:claim4eq4}. The same applies for the general case~\eqref{eq:claim4c} using the general form~\eqref{eq:claim4general}.

% \textbf{Remark:} If we choose $a=\cos{\theta}$ and $b=\ii \sin{\theta}$ for any angle $\theta$, then the term of order $\mathcal{O}(\epsilon)$ vanishes exactly due to the fact that $\bra{\alpha}\ket{-\alpha}\in \mathbb{R}$.

% \begin{proposition}\label{proposition2} Let $\kittyket = a \ket{\alpha} + b \ket{-\alpha}$ and $\braket{\alpha}{-\alpha} =\epsilon$, with $\epsilon\ll 1$. 
% Then
% \begin{align}
%     &\kittybra:\hat T_{\mu \nu}
%  \hat T'_{\sigma' \tau'}:\kittyket  \\ 
%     &\quad= \kittybra:\hat T_{\mu \nu}:\kittyket   \nonumber  \kittybra:\hat T'_{\sigma' \tau'}:\kittyket+\mathcal{O}(\epsilon)
% \end{align}
% \end{proposition}

% \textbf{Proof:}
% The proof is similar to the one for Proposition \ref{proposition1} once one makes use of Corollary \ref{corollary2}.

\begin{theorem} \label{thm:cat}
    In a globally hyperbolic stationary spacetime, for the cat state in equation \eqref{eq:catt}, the (normal-ordered) uncertainty of the stress-energy tensor, as well as all its higher order central moments (that we notate $\mu_n, n \ge 2$) vanish
    \begin{itemize}
        \item up to $\mathcal{O}(\epsilon)$, if  $\braket{\alpha}{-\alpha} =\epsilon$, or,
        \item exactly, if $a = \cos{\theta}$ and $b = \ii\sin{\theta}$, for $\theta \in [0,2\pi)$.
    \end{itemize}
\end{theorem}
\begin{proof}
Let us introduce the definition of the $n$th order moment $\mu_n$ of stress-energy tensor for some state $\ket{\psi}$ as follows:
\begin{align} \label{eq:nthmoment}
    \mu_{n}(x,x',\cdots,x^{(n)}) &= \sum^n_{m=0}  (-1)^{n-m} \nonumber \binom{n}{m}\\ 
    &\times \frac{1}{n!}\mathcal{P}\big[\tilde{\mu}_{m}(x,x',\cdots,x^{(m)}) \nonumber \\
    &\times\tilde{\mu}(x^{(m+1)})\times\cdots\times\tilde{\mu}(x^{(n)}) \big],
\end{align}
where 
\begin{eqnarray}
    \label{eq:higherMoment}
    \tilde{\mu}_1(x) &=& \bra{\psi}:\hat{T}_{\alpha \beta}:\ket{\psi}, \\
    \tilde{\mu}_2(x,x') &=& \bra{\psi}:\hat{T}_{\alpha \beta} \ \hat{T'}_{\alpha' \beta'}   :\ket{\psi},\\
    \tilde{\mu}_3(x,x',x'') &=& \bra{\psi}:\hat{T}_{\alpha \beta} \! \hat{T'}_{\alpha' \beta'}   \! \hat{T''}_{\alpha'' \beta''} \!\! :\ket{\psi}\!, \\
    \cdots \  &\cdots& \ \cdots \ \cdots\nonumber \\
    \tilde{\mu}_m(x,\cdots,x^{(m)}) &=& \bra{\psi}:\hat{T}_{\mu \nu}  \cdots \hat{T}^{(m)}_{\mu^{(m)} \nu^{(m)}} :\ket{\psi}\!,\label{eq:psi_mu_n}
\end{eqnarray}
$\tilde{\mu}_0=1$, $\tilde{\mu} = \tilde{\mu}_1 $, and the operator $\mathcal{P}$ produces all the possible permutations of the spacetime coordinates $x,x',x'', \cdots, x^{(n)}$ of the functions inside its argument and add them together. For example, 
\begin{align}\mathcal{P}[ O(x,x')  P(x'')] &=  O(x,x') P(x'')+  O(x',x) P(x'') \nonumber \\
&+ O(x,x'') P(x')+  O(x'',x) P(x') \nonumber \\
&+ O(x'',x') P(x)+  O(x',x'') P(x).
\end{align}
If we put $\ket{\psi} = \kittyket$ in equation \eqref{eq:psi_mu_n}, then
    \begin{align} \label{eq:hat_mu_m}
        \tilde{\mu}_m(x,x',\cdots,x^{(m)}) &= \kittybra:\hat{T}_{\mu \nu}  \cdots \hat{T}^{(m)}_{\mu^{(m)} \nu^{(m)}} :\kittyket.
    \end{align}
    Applying Proposition \ref{proposition1}, we obtain
    \begin{align}
        \tilde{\mu}_m(x,x',\cdots,x^{(m)}) = 
        & \bra{\alpha}:\hat{T}_{\rho \sigma}  \cdots \hat{T}^{(m)}_{\rho^{(m)} \sigma^{(m)}} \!\! :\ket{\alpha} \nonumber \\ &+\mathcal{O}(\epsilon).
    \end{align}
    Now, using the notation in equation \eqref{eq:notationTmunu} and according to Theorem \ref{thm:poly}, 
    \begin{align}
        \bra{\alpha}:&\hat{T}_{\rho \sigma}  \cdots  \nonumber \hat{T}^{(m)}_{\rho^{(m)} \sigma^{(m)}} \!\! :\ket{\alpha} \\
        & = {T}_{\rho \sigma}[\alpha,\alpha] \times \cdots \times{T}^{(m)}_{\rho^{(m)} \sigma^{(m)}} [\alpha,\alpha] \\
        & = \bra{\alpha}:\hat{T}_{\rho \sigma}  :\ket{\alpha} \times \cdots  \times \bra{\alpha}: \hat{T}^{(m)}_{\rho^{(m)} \sigma^{(m)}} \!\! :\ket{\alpha} \\
        & = \kittybra:\hat{T}_{\rho \sigma}  :\kittyket \times \cdots  \times \kittybra: \hat{T}^{(m)}_{\rho^{(m)} \sigma^{(m)}} \!\! :\kittyket \nonumber \\
        &\quad \quad +\mathcal{O}(\epsilon).
    \end{align}
    
    In the last equation above, we have used the formula from equation \eqref{eq:claim4}. Now, by recalling the definition of the first order moment from equation \eqref{eq:higherMoment} and combining with equation \eqref{eq:hat_mu_m}, we get
    \begin{align} \label{eq:mu_m_mu_mu}
        \tilde{\mu}_m(x,x',\cdots,x^{(m)}) = 
        & \ \tilde{\mu}(x) \times \tilde\mu(x') \times \cdots \tilde{\mu}(x^{(m)}) \nonumber \\ &\quad \quad +\mathcal{O}(\epsilon).
    \end{align}
    This implies
    \begin{align}
    \mathcal{P}\big[\tilde{\mu}_m(x,x',\cdots,x^{(m)})] = 
        & \ m! \times \tilde{\mu}(x) \times \tilde\mu(x') \times \cdots \tilde{\mu}(x^{(m)}) \nonumber \\ &\quad \quad +\mathcal{O}(\epsilon).
    \end{align}
    Next, by plugging this to the definition of $n$th order moment of stress-energy tensor defined in equation \eqref{eq:nthmoment}, 
    \begin{align}
        \mu_{n}(x,x',\cdots,x^{(n)}) &= \sum^n_{m=0} \binom{n}{m} (-1)^{n-m} \nonumber \\ 
        &\times \tilde{\mu}(x) \times \cdots \tilde{\mu}(x^{(m)}) \times\tilde{\mu}(x^{(m+1)}) \nonumber \\
        &\times\cdots\times\tilde{\mu}(x^{(n)})+ \mathcal{O}(\epsilon) \nonumber\\
        &=\mathcal{O}(\epsilon). 
    \end{align}
    
    % For cat states, if we can show that $\Delta_{\mu\nu\rho'\sigma'}(\mathsf x, \mathsf x')$ in the equation \ref{eq:KuoFord} is $\mathcal{O}(\epsilon)$, we are done with the first case. Since we are assuming $||\alpha||\neq 0$, we don't have to worry about the denominator. From the proposition \ref{proposition1} and the theorem \ref{thm:poly}, we have
    % \begin{align}
    %     \kittybra:\hat T_{\mu \nu} \hat T'_{\sigma' \rho'}:\kittyket &= \bra{\alpha}:\hat T_{\mu \nu}\hat T'_{\sigma' \rho'}:\ket{\alpha} + \mathcal{O}(\epsilon)\\
    %     &= \bra{\alpha}:\hat T_{\mu \nu} :\ket{\alpha} \nonumber \\
    %     &\quad\times\bra{\alpha}:\hat T'_{\sigma' \rho'}:\ket{\alpha} +\mathcal{O}(\epsilon) \\
    %     &=\kittybra:\hat T_{\mu \nu} :\kittyket \nonumber \\&\quad\times\kittybra:\hat T'_{\sigma' \rho'}:\kittyket+\mathcal{O}(\epsilon).
    %     \label{eq:thmcatTTcat}
    % \end{align}
    % The last equation shows that $\Delta_{\mu\nu\rho'\sigma'}(\mathsf x, \mathsf x') \ll 1$.
    This concludes the first part of the proof for the theorem. However, we can do better than this if we choose $a = \cos{\theta}$ and $b = \ii \sin{\theta}$. We need to check the normalization condition $\braket{\raisebox{-.43ex}{\SchrodingersCat{1}}\!\!}{\raisebox{-.43ex}{\SchrodingersCat{1}}\!\!} = 1$, by noticing that
    \begin{align}
        \braket{\raisebox{-.43ex}{\SchrodingersCat{1}}\!\!}{\raisebox{-.43ex}{\SchrodingersCat{1}}\!\!} &= \cos^2{\theta} \braket{\alpha}{\alpha} + \sin^2{\theta} \braket{-\alpha}{-\alpha} \nonumber \\ &- \ii  \sin{\theta} \cos{\theta} \braket{\alpha}{-\alpha} + \ii  \sin{\theta} \cos{\theta} \braket{-\alpha}{\alpha}, \label{eq:normalkitty}
    \end{align}
    $\braket{\alpha}{\alpha}=\braket{-\alpha}{-\alpha}$, and $\braket{-\alpha}{\alpha}=\braket{\alpha}{-\alpha}$. The most important part here is that $\mathcal{O}(\epsilon)$ doesn't appear in the equation \eqref{eq:normalkitty} because of the specific choice of coefficients. This can also be applied to Proposition \ref{proposition1}, that is, for this specific choice of coefficients,
    \begin{equation}
        \kittybra: P[\hat T_{\mu \nu}]:\kittyket = \bra{\alpha}:\hat P[T_{\mu \nu}]:\ket{\alpha}.
    \end{equation}
    In similar fashion of the first part of the proof, it can be derived that $\mu_{n} = 0$ when $\kittyket=\cos{\theta}\ket{\alpha}+\ii \sin{\theta}\ket{-\alpha}$. 
\end{proof}

\begin{corollary}
    In a globally hyperbolic stationary spacetime, the cat state under the conditions mentioned in Theorem \ref{thm:cat} satisfies the Kuo-Ford criterion.
\end{corollary}
\begin{proof}
    The proof is trivial by using the special case $n=2$ of Theorem \ref{thm:cat} and equation \eqref{eq:KuoFord}.
\end{proof}

% \begin{lemma} \label{lemma:mu}
%     If $\ket{\psi}$ is coherent state or cat states with condition mentioned in proposition \ref{proposition1} and $n \ge 2$,
%     \begin{equation}
%         \hat{\mu}_n(x,x',\cdots,x^{(n)}) = \hat{\mu}(x)\times\hat{\mu}(x')\times\cdots\times\hat{\mu}(x^{(n)}).
%     \end{equation}
% \end{lemma}
% \begin{proof}
%     This proof is just a direct consequence of the theorem \ref{thm:poly} and proposition \ref{proposition1}.
% \end{proof}
% \begin{proposition}
%     The higher order moments $\mu_n (n \ge 2)$   of energy-momentum tensor vanishes for coherent states and cat states with conditions mentioned in proposition \ref{proposition1}.
% \end{proposition}
% \begin{proof}
%     In equation \ref{eq:higherMoment} by plugging in result from lemma \ref{lemma:mu}, we can easily derive that $\mu_n(x,x',\cdots,x^{(n)}) = 0$ for $n\ge2$.
% \end{proof}

\section{\label{sec:conclusions} Conclusions}

We have shown that cat states (which can be seen as the quantum superpositions of two classical states) can gravitate semiclassically. This happens in two different regimes:
\begin{enumerate}
\item When we consider large cat states.---Namely, when the coherent amplitude of the cat state is large enough (so the inner product between the two terms in the superposition is small);
\item For any cat state for which the coefficients of the supersposition are such that the relative phase between the two coherent states is $\pi/2$.
\end{enumerate}

More specifically, we have shown that the uncertaintly---as well as all the symmetrized (regularized) central higher moments--- of the stress-energy density for the cat state either become negligible (for case 1) or vanish exactly (for case 2), despite the fact that cat states are not Gaussian. 

According to the Kuo-Ford criteria~\cite{Kuo93}, this means that, the cat states in cases 1 and 2, despite displaying genuine quantum behaviour, are within the scope of validity of semiclassical gravity. Notice that this statement relies heavily on the coherent nature of the two superimposed coherent states forming a cat state. This means that this result does not apply to superpositions of macroscopic classical states states other than coherent states.

\begin{acknowledgments}
The authors thank Albert Roura for very insightful discussions during RQI-N 2023. Research at Perimeter Institute is supported in part by the Government of Canada through the Department of Innovation, Science and Economic Development Canada and by the Province of Ontario through the Ministry of Colleges and Universities. E. M-M. is funded by the NSERC Discovery program as well as his Ontario Early Researcher Award.
\end{acknowledgments}

\appendix

\section{Proof of equation \eqref{eqn:ccrelation}}\label{appendixA}

% Here we will show equation \eqref{eqn:ccrelation} in the continuous form. Notice that the discrete version of the proof is analogous.

% \textbf{Claim}
% \begin{eqnarray}
%     [\hat{a}_{\bm{k}}, \hat D(\alpha)] = \alpha(\bm{k}) \hat D(\alpha).
% \end{eqnarray}

% \textbf{Proof:}
% \begin{eqnarray}
%     [\hat{a}_{\bm{k}}, \hat D(\alpha)] 
%     &=&\left[\hat{a}_{\bm{k}}, \sum_{n=0}^{\infty} \frac{1}{n!}\left(\int \dd^d\bm{p} \  (\alpha(\bm{p})\hat{a}^{\dagger}_{\bm{p}}-\alpha^*(\bm{p})\hat{a}_{\bm{p}})\right)^n\right] \nonumber \\
%     &=& \sum_{n=0}^{\infty} \frac{n}{n!} \left[\hat{a}_{\bm{k}} , \int \dd^d\bm{p} \  (\alpha(\bm{p})\hat{a}^{\dagger}_{\bm{p}}-\alpha^*(\bm{p})\hat{a}_{\bm{p}})\right]  \nonumber \\ 
%     && \left(\int \dd^d\bm{p} \  (\alpha(\bm{p})\hat{a}^{\dagger}_{\bm{p}}-\alpha^*(\bm{p})\hat{a}_{\bm{p}})\right)^{n-1} \nonumber \\
%     &=& \sum_{n=0}^{\infty} \frac{n}{n!} \alpha(\bm{k}) \left(\int d^d\bm{p} \  (\alpha(\bm{p})\hat{a}^{\dagger}_{\bm{p}}-\alpha^*(\bm{p})\hat{a}_{\bm{p}})\right)^{n-1} \nonumber \\
%     &=& \alpha(\bm{k}) \hat D(\alpha). \nonumber
% \end{eqnarray}

\textbf{Claim}
\begin{eqnarray}
    [\hat{a}_i, \hat D(\alpha)] = \alpha_i \hat D(\alpha).
\end{eqnarray}

\textbf{Proof:}
\begin{eqnarray}
    [\hat{a}_i, \hat D(\alpha)] 
    &=&\left[\hat{a}_i, \sum_{n=0}^{\infty} \frac{1}{n!}\left(\sumint_j \  (\alpha_j\hat{a}^{\dagger}_j-\alpha^*_j\hat{a}_j)\right)^n\right] \nonumber \\
    &=& \sum_{n=1}^{\infty} \frac{n}{n!} \left[\hat{a}_i , \left(\sumint_k \  (\alpha_k\hat{a}^{\dagger}_k-\alpha^*_k\hat{a}_k)\right)\right]  \nonumber \\ 
    && \quad \quad \quad \times \left(\sumint_j \  (\alpha_j\hat{a}^{\dagger}_j-\alpha^*_j\hat{a}_j)\right)^{n-1} \nonumber \\
    &=& \sum_{n=1}^{\infty} \frac{n}{n!} \alpha_i \left(\sumint_j \  (\alpha_j\hat{a}^{\dagger}_j-\alpha^*_j\hat{a}_j)\right)^{n-1} \nonumber \\
    &=& \alpha_i \hat D(\alpha). \nonumber
\end{eqnarray}

\bibliography{references}% Produces the bibliography via BibTeX.

\end{document}